\documentclass[aap,preprint]{imsart}

\RequirePackage[utf8]{inputenc}
\RequirePackage[OT1]{fontenc}
\RequirePackage{etex}
\RequirePackage{amsthm,amsmath}
\RequirePackage[numbers]{natbib}
\RequirePackage[colorlinks,citecolor=blue,urlcolor=blue]{hyperref}
\RequirePackage{etoolbox}
\RequirePackage{xparse}
\RequirePackage{ifthen}
\RequirePackage{mathtools}
\RequirePackage[capitalize]{cleveref}
\RequirePackage{dsfont}
\RequirePackage{amssymb}
\RequirePackage{tikz}
\RequirePackage{pgfplots}
\RequirePackage{bookmark}
\RequirePackage{ulem}
\RequirePackage{subcaption}
\RequirePackage[hide=true,dir=bashful]{bashful}

\makeatletter
\ifdef{\Cref}{%
\@ifpackageloaded{autonum}{%
\autonum@generatePatchedReferenceCSL{Cref}%
}{}}{}%
\makeatother

\arxiv{1604.02109}

\include{git-include}

\startlocaldefs
\newcounter{todo}
\newcommand{\todo}[1][\empty]{
  \bookmarksetupnext{keeplevel}  
  \subpdfbookmark{TODO}{todo.\arabic{todo}}
  \bookmarksetup{keeplevel=false}  
  \refstepcounter{todo}
  \textbf{\textcolor{red}{\ifthenelse{\equal{#1}{\empty}}{TODO}{TODO:~\emph{#1}}}}
}

\newcommand{\declarecommand}[1]{\providecommand{#1}{}\renewcommand{#1}}
\declarecommand{\arxiv}[1]{%
  arXiv:~\href{https://arxiv.org/abs/#1}{#1}%
}

\let\originalleft\left
\let\originalright\right
\let\originalmiddle\middle
\renewcommand{\left}{\mathopen{}\mathclose\bgroup\originalleft}
\renewcommand{\right}{\aftergroup\egroup\originalright}
\renewcommand{\middle}{\originalmiddle}

\newcommand{\Base}{\ensuremath{\{-1,1\}}}
\newcommand{\bvt}[1]{\ensuremath{\chi_{#1}}}

\newcommand{\Tr}{\ensuremath{T}}

\DeclareDocumentCommand{\pair}{ O{\empty} m m}{
  \ensuremath{%
    \ifthenelse{\equal{#1}{big}}{\big\langle #2, #3 \big\rangle}{
      \ifthenelse{\equal{#1}{Big}}{\Big\langle #2, #3 \Big\rangle}{
        \ifthenelse{\equal{#1}{bigg}}{\bigg\langle #2, #3 \bigg\rangle}{
          \ifthenelse{\equal{#1}{Bigg}}{\Bigg\langle #2, #3 \Bigg\rangle}{
            \ifthenelse{\equal{#1}{normal}}{\langle #2, #3 \rangle}{\left\langle #2, #3 \right\rangle}
          }
        }
      }
    }
  }
}

\newcommand{\sgn}{\ensuremath{\mathrm{sgn}}}

\newcommand{\ot}{\ensuremath{\bar t}}
\newcommand{\op}{\ensuremath{\bar p}}

\newcommand{\orho}{\ensuremath{\bar \rho}}

\newcommand{\defas}{\ensuremath{\vcentcolon=}}
\newcommand{\defasAlign}{\vcentcolon\!\!&=}

\newcommand{\vt}[1]{\ensuremath{\boldsymbol{#1}}}
\newcommand{\Nto}[1]{\ensuremath{\{1,2,\dots,#1\}}}
\newcommand{\Ntoo}[1]{\ensuremath{[#1]}}
\newcommand{\RR}{\ensuremath{\mathbb{R}}}
\newcommand{\NN}{\ensuremath{\mathbb{N}}}

\newcommand{\set}[1]{\ensuremath{\mathcal{#1}}}

\newcommand{\III}{\ensuremath{\set{I}}}
\newcommand{\NNN}{\ensuremath{\set{N}}}
\newcommand{\PPP}{\ensuremath{\set{P}}}

\newcommand{\fOrig}{\ensuremath{f}}
\newcommand{\gOrig}{\ensuremath{g}}
\newcommand{\fFourier}[1]{\ensuremath{\hat \fOrig_{#1}}}
\newcommand{\gFourier}[1]{\ensuremath{\hat \gOrig_{#1}}}
\newcommand{\FktI}{\ensuremath{\phi}}
\newcommand{\FktIb}{\ensuremath{\FktI}}

\newcommand{\FktII}{\ensuremath{\psi}}
\newcommand{\FktIII}{\ensuremath{\FktII}}
\newcommand{\FktV}{\ensuremath{\FktIII}}

\newcommand{\FktIV}{\ensuremath{\xi}}

\newcommand{\FktVII}{\ensuremath{\gamma}}

\newcommand{\FktLemII}{\ensuremath{f}}

\newcommand{\mytheta}[1]{\ensuremath{\theta_{#1}}}
\newcommand{\thetarho}{\ensuremath{\mytheta{\rho}}}
\newcommand{\thetaplus}{\ensuremath{\thetarho^+}}
\newcommand{\thetaminus}{\ensuremath{\thetarho^-}}

\newcommand{\ai}{\ensuremath{a}}
\newcommand{\oai}{\ensuremath{\bar{\ai}}}
\newcommand{\bi}{\ensuremath{b}}
\newcommand{\obi}{\ensuremath{\bar{\bi}}}

\newcommand{\aii}{\ensuremath{\alpha}}
\newcommand{\oaii}{\ensuremath{\bar{\aii}}}
\newcommand{\bii}{\ensuremath{\beta}}
\newcommand{\obii}{\ensuremath{\bar{\bii}}}

\newcommand{\cc}{\ensuremath{c}}
\newcommand{\xx}{\ensuremath{x}}

\newcommand{\CST}[1][\empty]{\ensuremath{C_{\ifthenelse{\equal{#1}{\empty}}{\aii,\bii}{#1}}}}

\newcommand{\rhomin}{\ensuremath{\rho_\circ}}
\newcommand{\rhommin}{\ensuremath{\rho_-}}
\newcommand{\rhomax}{\ensuremath{\rho_+}}

\newcommand{\tauplus}{\ensuremath{\tau^+}}
\newcommand{\tauminus}{\ensuremath{\tau^-}}
\newcommand{\rhoI}{\ensuremath{\rho^+}}
\newcommand{\rhoII}{\ensuremath{\rho^-}}
\newcommand{\setS}{\ensuremath{\mathcal{S}}}

\newcommand{\ti}{\ensuremath{u}}
\newcommand{\tii}{\ensuremath{\ti^*}}
\newcommand{\tI}{\ensuremath{\ti_1}}
\newcommand{\tII}{\ensuremath{\ti_2}}

\newcommand{\setI}{\ensuremath{I}}
\newcommand{\setII}{\ensuremath{U}}
\newcommand{\setIIb}{\ensuremath{V}}

\newcommand{\rv}[1]{\ensuremath{\mathsf{\uppercase{#1}}}}
\newcommand{\rvt}[1]{\vt{\rv{#1}}} 
\newcommand{\sCref}[2]{{\Cref{#2} of~\cref{#1}}}
\newcommand{\scref}[2]{{\cref{#2} of~\cref{#1}}}

\DeclareDocumentCommand{\mutInf}{ O{\empty} m m}{
  \ensuremath{\mathrm{I}
    \ifthenelse{\equal{#1}{big}}{\big({#2; #3}\big)}{
      \ifthenelse{\equal{#1}{Big}}{\Big({#2; #3}\Big)}{
        \ifthenelse{\equal{#1}{bigg}}{\bigg({#2; #3}\bigg)}{
          \ifthenelse{\equal{#1}{Bigg}}{\Bigg({#2; #3}\Bigg)}{
            \ifthenelse{\equal{#1}{normal}}{({#2; #3})}{\left({#2; #3}\right)}
          }
        }
      }
    }
  }
}

\DeclareDocumentCommand{\binEnt}{ O{\empty} m}{
  \ensuremath{\mathrm{H}
    \ifthenelse{\equal{#1}{big}}{\big({#2}\big)}{
      \ifthenelse{\equal{#1}{Big}}{\Big({#2}\Big)}{
        \ifthenelse{\equal{#1}{bigg}}{\bigg({#2}\bigg)}{
          \ifthenelse{\equal{#1}{Bigg}}{\Bigg({#2}\Bigg)}{
            \ifthenelse{\equal{#1}{normal}}{({#2})}{\left(#2\right)}
          }
        }
      }
    }
  }
}

\DeclareDocumentCommand{\Prob}{ O{\empty} O{\empty} m}{
  \ensuremath{\mathrm{P}\ifthenelse{\equal{#1}{\empty}}{}{_{#1}}
    \ifthenelse{\equal{#2}{big}}{\big\{#3\big\}}{
      \ifthenelse{\equal{#2}{Big}}{\Big\{#3\Big\}}{
        \ifthenelse{\equal{#2}{bigg}}{\bigg\{#3\bigg\}}{
          \ifthenelse{\equal{#2}{Bigg}}{\Bigg\{#3\Bigg\}}{
            \ifthenelse{\equal{#2}{normal}}{\{#3\}}{\left\{#3\right\}}
          }
        }
      }
    }
  }
}

\DeclareDocumentCommand{\Exp}{ O{\empty} m}{
  \ensuremath{\mathrm{\mathds{E}}
    \ifthenelse{\equal{#1}{big}}{\big[ #2 \big]}{
      \ifthenelse{\equal{#1}{Big}}{\Big[ #2 \Big]}{
        \ifthenelse{\equal{#1}{bigg}}{\bigg[ #2 \bigg]}{
          \ifthenelse{\equal{#1}{Bigg}}{\Bigg[ #2 \Bigg]}{
            \ifthenelse{\equal{#1}{normal}}{[ #2 ]}{\left[ #2 \right]}
          }
        }
      }
    }
  }
}

\DeclareDocumentCommand{\condExp}{ O{\empty} m m}{
  \ensuremath{\mathds{E}
    \ifthenelse{\equal{#1}{big}}{\big[{#2\big| #3}\big]}{
      \ifthenelse{\equal{#1}{Big}}{\Big[{#2\Big| #3}\Big]}{
        \ifthenelse{\equal{#1}{bigg}}{\bigg[{#2\bigg| #3}\bigg]}{
          \ifthenelse{\equal{#1}{Bigg}}{\Bigg[{#2\Bigg| #3}\Bigg]}{
            \ifthenelse{\equal{#1}{normal}}{[{#2| #3}]}{\left[#2 \middle| #3\right]}
          }
        }
      }
    }
  }
}

\DeclareDocumentCommand{\card}{ O{\empty} m}{
  \ensuremath{
    \ifthenelse{\equal{#1}{big}}{\big\vert #2 \big\vert}{
      \ifthenelse{\equal{#1}{Big}}{\Big\vert #2 \Big\vert}{
        \ifthenelse{\equal{#1}{bigg}}{\bigg\vert #2 \bigg\vert}{
          \ifthenelse{\equal{#1}{Bigg}}{\Bigg\vert #2 \Bigg\vert}{
            \ifthenelse{\equal{#1}{normal}}{\vert #2 \vert}{\left\vert #2 \right\vert}
          }
        }
      }
    }
  }
}

\DeclareDocumentCommand{\Pcond}{ O{\empty} O{\empty} m m}{
  \ensuremath{\mathrm{P}
    \ifthenelse{\equal{#1}{\empty}}{}{_{#1|#2}}
    \ifthenelse{\equal{#3#4}{}}{}{\left\{{#3}{}\middle\vert{}{#4}\right\}}}
}

\DeclareDocumentCommand{\ent}{ O{\empty} m}{
  \ensuremath{
    \ensuremath{\mathrm{H}
      \ifthenelse{\equal{#1}{big}}{\big(#2\big)}{
        \ifthenelse{\equal{#1}{Big}}{\Big(#2\Big)}{
          \ifthenelse{\equal{#1}{bigg}}{\bigg(#2\bigg)}{
            \ifthenelse{\equal{#1}{Bigg}}{\Bigg(#2\Bigg)}{
              \ifthenelse{\equal{#1}{normal}}{(#2)}{\left(#2\right)}
            }
          }
        }
      }
    }
  }
}
\DeclareDocumentCommand{\condEnt}{ O{\empty} m m}{
  \ensuremath{
    \ensuremath{\mathrm{H}
      \ifthenelse{\equal{#1}{big}}{\big({#2\big| #3}\big)}{
        \ifthenelse{\equal{#1}{Big}}{\Big({#2\Big| #3}\Big)}{
          \ifthenelse{\equal{#1}{bigg}}{\bigg({#2\bigg| #3}\bigg)}{
            \ifthenelse{\equal{#1}{Bigg}}{\Bigg({#2\Bigg| #3}\Bigg)}{
              \ifthenelse{\equal{#1}{normal}}{({#2| #3})}{\left(#2 \middle| #3\right)}
            }
          }
        }
      }
    }
  }
}

\theoremstyle{plain}
\newtheorem{theorem}{Theorem}\newtheorem{lemma}{Lemma}\newtheorem{conjecture}{Conjecture}\newtheorem{proposition}{Proposition}
\theoremstyle{remark}

\crefname{enumi}{part}{parts}
\Crefname{enumi}{Part}{Parts}

\crefname{equation}{}{}
\Crefname{equation}{}{}

\crefname{figure}{figure}{figures}
\Crefname{figure}{Figure}{Figures}

\crefname{case}{case}{cases}
\Crefname{case}{Case}{Cases}

\pgfplotsset{sketch style/.append style={%
    axis x line=middle,%
    axis y line=middle,%
    xlabel={$\rho$},%
    ylabel={},%
    x label style={anchor=north},%
    legend style={%
      at={(0.02,0.98)},%
      anchor=north west},%
    width=7.3cm,%
    height=7.3cm,%
  }%
}
\definecolor{myred}{RGB}{153, 0, 0}
\definecolor{myblue}{RGB}{0, 0, 153}

\newcommand{\prm}{^\prime}
\newcommand{\pprm}{^{\prime\prime}}

\endlocaldefs

\allowdisplaybreaks

\begin{document}

\begin{frontmatter}
\title{Dictator Functions Maximize Mutual Information}
\runtitle{Dictator Functions Maximize Mutual Information}

\begin{aug}
\author{\fnms{Georg} \snm{Pichler}%
  \thanksref{t1}   \ead[label=e1]{georg.pichler@gmail.com}},
\author{\fnms{Pablo} \snm{Piantanida}\ead[label=e2]{pablo.piantanida@centralesupelec.fr}}
\and
\author{\fnms{Gerald} \snm{Matz}%
  \thanksref{t1}   \ead[label=e3]{gerald.matz@nt.tuwien.ac.at}
}

\runauthor{G. Pichler et al.}

\affiliation{TU Wien and CentraleSup\'elec-CNRS-Universit\'e Paris-Sud}

\thankstext{t1}{Supported by WWTF Grants ICT12-054 and ICT15-119.} 

\address{Institute of Telecommunications \\
Technische Universität Wien \\
Gusshausstraße 25 / E389 \\
1040 Vienna, Austria \\
\printead{e1}\\
\phantom{E-mail:\ }\printead*{e3}}

\address{CentraleSup\'elec-CNRS-Universit\'e Paris-Sud \\
3 rue Joliot-Curie \\
F-91192 Gif-sur-Yvette Cedex, France \\
\printead{e2}}
\end{aug}

\begin{abstract}
Let $(\rvt X, \rvt Y)$ denote $n$ independent, identically distributed copies of 
two arbitrarily correlated Rademacher random variables $(\rv x, \rv y)$.
We prove that the inequality $\mutInf{\fOrig(\rvt X)}{\gOrig(\rvt Y)} \le \mutInf{\rv x}{\rv y}$ holds
for any two Boolean functions: $\fOrig,\gOrig \colon \{-1,1\}^n \to \{-1,1\}$ ($\mutInf{\,\cdot\,}{\cdot}$ denotes mutual information).
We further show that equality in general is achieved only by the dictator 
functions $\fOrig(\vt x)=\pm \gOrig(\vt x)=\pm x_i$, $i \in \Nto{n}$.
\end{abstract}



\end{frontmatter}

\section{Introduction and Main Results}

Let
$(\rv x, \rv y)$ be two dependent Rademacher random variables on $\{-1,1\}$, with correlation coefficient
$\rho \defas \Exp{\rv x\rv y} \in [-1,1]$.
For given $n \in \NN$,
let $(\rvt x, \rvt y) = (\rv x, \rv y)^n$ be $n$ independent, identically distributed copies of $(\rv x, \rv y)$.
We will use the notation from \cite{Cover2006Elements} for information-theoretic quantities.
In particular, $\Exp{\rv x}$, $\ent{\rv x}$, and $\mutInf{\rv x}{\rv y}$ denote expectation, entropy, and mutual information, respectively.
Motivated by problems in computational biology \cite{Klotz2014Canalizing}, Kumar and Courtade formulated the following conjecture \cite[Conjecture~1]{Kumar2013Which}.
\begin{conjecture}
  \label{conj:full_conjecture}
  For any Boolean function $\fOrig \colon \{-1,1\}^n \to \{-1,1\}$,
  \begin{align}
    \label{eq:full_conjecture}
    \mutInf[big]{\fOrig(\rvt x)}{\rvt y} \le \mutInf{\rv x}{\rv y} .
  \end{align}
\end{conjecture}
\noindent
This claim -- while seemingly innocent at first sight -- has received significant interest and resisted several efforts to find a proof
(see the discussion in \cite[Section~IV]{Courtade2014Which}).
Note that $\fOrig=\bvt{i}$ for any dictator function~\cite[Definition 2.3]{ODonnell2014Analysis} $\bvt{i}(\vt x) \defas x_i$, $i \in \Nto{n}$ achieves equality in \cref{eq:full_conjecture}.

We next state the main result of this paper, which is a relaxed version of \cref{conj:full_conjecture}, involving two Boolean functions.

\begin{theorem}
  \label{thm:main}
  For any two Boolean functions $\fOrig,\gOrig \colon \{-1,1\}^n \to \{-1,1\}$,
  \begin{align}
    \label{eq:thm:main}
    \mutInf[big]{\fOrig(\rvt x)}{\gOrig(\rvt y)} \le \mutInf{\rv x}{\rv y} .
  \end{align}
\end{theorem}
\noindent
If \cref{eq:full_conjecture} were true, this statement would readily follow from the data processing inequality \cite[Theorem~2.8.1]{Cover2006Elements}.
\Cref{thm:main} was stated as an open problem in \cite{Courtade2014Which} and \cite[Section~IV]{Kumar2013Which}, and separately investigated in \cite{Anantharam2013hypercontractivity}.
A proof of \eqref{eq:thm:main} was previously available only under the additional restrictive assumptions that $\fOrig$ and $\gOrig$ are equally biased (i.e., $\Exp{\fOrig(\rvt x)} = \Exp{\gOrig(\rvt x)}$) and satisfy the condition
\begin{align}
  \Prob{\fOrig(\rvt x) = 1, \gOrig(\rvt x) = 1} \ge \Prob{\fOrig(\rvt x) = 1} \Prob{\gOrig(\rvt x) = 1} . \label{eq:correlation_condition}
\end{align}
The reader is invited to see \cite[Section~IV]{Courtade2014Which} for further details. 
In this paper, we use Fourier-analytic tools to prove \cref{thm:main} without any additional restrictions on $\fOrig$ and $\gOrig$. 
We suitably bound the Fourier coefficients of $\fOrig$ and $\gOrig$, and thereby reduce \cref{eq:thm:main} to an elementary inequality, which is subsequently established.

A careful inspection of the proof of \cref{thm:main} reveals that in general, up to sign changes, the 
dictator functions $\bvt{i}$, $i \in \Nto{n}$ are the unique maximizers of $\mutInf[big]{\fOrig(\rvt x)}{\gOrig(\rvt y)}$.
\begin{proposition}
  \label{pro:uniqueness}
  If $0<|\rho|<1$, equality in \cref{eq:thm:main} is achieved if and only if
  $\fOrig=\pm \gOrig = \pm\bvt{i}$ for some $i \in \Nto{n}$.
\end{proposition}

\section{Proof of \texorpdfstring{\cref{thm:main}}{Theorem \ref{thm:main}}}
\label{sec:main-result}

Define $\Ntoo{n} \defas \Nto{n}$ and let $f$, $g$ be two Boolean functions on the Boolean hypercube, i.e., $\fOrig,\gOrig \colon \Base^n \to \Base$. Denote their Fourier expansions (cf.~\cite[(1.6)]{ODonnell2014Analysis}) $\fOrig(\vt x) = \sum_{\setS \subseteq \Ntoo{n}} \fFourier{\setS} \bvt{\setS}(\vt x)$ and $\gOrig(\vt x) = \sum_{\setS \subseteq \Ntoo{n}} \gFourier{\setS} \bvt{\setS}(\vt x)$, using the basis $\bvt{\setS}(\vt x) \defas \prod_{i \in \setS} x_i$ for $\setS \subseteq \Ntoo{n}$.
Define $\ai \defas \frac{1 + \fFourier{\varnothing}}{2} = \Prob{\fOrig(\rvt x) = 1}$, $\bi \defas \frac{1 + \gFourier{\varnothing}}{2} = \Prob{\gOrig(\rvt x) = 1}$ and $\thetarho \defas \frac 14 \sum_{\setS : \card{\setS} \ge 1} \fFourier{\setS} \gFourier{\setS} \rho^{\card{\setS}}$.
Without loss of generality, we may assume $\frac 12 \le \ai \le \bi \le 1$ and $\rho \in [0,1]$, as mutual information is symmetric and we have, with $\rvt y^* \defas \sgn(\rho)\rvt y$,
\begin{align}
  \mutInf[big]{\fOrig(\rvt x)}{\gOrig(\rvt y)} = \mutInf[big]{\sgn(\fFourier{\varnothing})\fOrig(\rvt x)}{\sgn(\gFourier{\varnothing})\gOrig(\sgn(\rho)\rvt y^*)} . \label{eq:symmetries}
\end{align}
In analogy to \cite[Proposition~1.9]{ODonnell2014Analysis}, the inner product satisfies
\begin{align}
  \pair{\fOrig}{\Tr_\rho \gOrig} &= \Exp{\fOrig(\rvt x)\gOrig(\rvt y)} = \fFourier{\varnothing} \gFourier{\varnothing} + 4 \thetarho = 1 - 2\Prob{\fOrig(\rvt x) \neq \gOrig(\rvt y)} , \label{eq:pair_prob}
\end{align}
where $\Tr_\rho$ is the noise operator \cite[Definition 2.46]{ODonnell2014Analysis}.
Defining $\ot \defas 1-t$ for a generic $t$, we can express the probabilities
\begin{align}
  \Prob{\fOrig(\rvt x) = 1, \gOrig(\rvt y) = -1} &= \ai\obi - \thetarho, \quad \Prob{\fOrig(\rvt x) = \gOrig(\rvt y) = 1} = \ai\bi + \thetarho ,\label{eq:joint_prob1}\\
  \Prob{\fOrig(\rvt x) = -1, \gOrig(\rvt y) = 1} &= \oai\bi - \thetarho, \; \Prob{\fOrig(\rvt x) = \gOrig(\rvt y) = -1} = \oai\obi + \thetarho . \label{eq:joint_prob2}
\end{align}
Using \cref{eq:joint_prob1}, \cref{eq:joint_prob2} and fundamental properties of mutual information~\cite[Section~2.4]{Cover2006Elements}, we obtain $\mutInf[big]{\fOrig(\rvt x)}{\gOrig(\rvt y)} = \FktIV(\thetarho,\ai,\bi)$ with
\begin{align}
  \FktIV(\mytheta{},\ai,\bi) \defas&\, \binEnt{\ai} + \binEnt{\bi} - \ent{\ai\bi + \mytheta{}, \ai\obi - \mytheta{}, \oai \bi - \mytheta{}, \oai \obi + \mytheta{}} ,
\end{align}
where, slightly abusing notation, we defined the binary entropy function $\binEnt{p} \defas \ent{p,\op}$ and $\ent[big]{(p_i)_{i \in \III}} \defas -\sum_{i \in \III} p_i \log_2 p_i$ for $\card{\III} > 1$.
By the non-negativity of probabilities \cref{eq:joint_prob1,eq:joint_prob2}, for any $\rho \in [0,1]$,
\begin{align}
  \label{eq:theta_rho_bound1}
  -\oai\obi \le \thetarho \le \ai\obi .
\end{align}
With $\PPP \defas \{\setS \subseteq \Ntoo{n} : \fFourier{\setS} \gFourier{\setS} > 0\} \setminus \{\varnothing\}$ and $\NNN \defas \{\setS \subseteq \Ntoo{n} : \fFourier{\setS} \gFourier{\setS} < 0\}$, we define
\begin{align}
  \tauplus &\defas \frac 14 \sum_{\setS \in \PPP} \fFourier{\setS} \gFourier{\setS} ,  &
  \tauminus &\defas \frac 14 \sum_{\setS \in \NNN} \fFourier{\setS} \gFourier{\setS}  \label{eq:thetas}
\end{align}
and apply the Schwarz inequality to show
\begin{align}
  \tauplus - \tauminus &= \frac 14 \sum_{\setS:\card{\setS} \ge 1} |\fFourier{\setS}| |\gFourier{\setS}| \\
                      &\le \frac 14 \sqrt{(1-\fFourier{\varnothing}^2)(1-\gFourier{\varnothing}^2)}                       = \sqrt{\ai \oai \bi \obi} . \label{eq:tau_diff_bound}
\end{align}
As $\mytheta{1} = \tauplus + \tauminus$, we combine \cref{eq:theta_rho_bound1,eq:tau_diff_bound} to obtain
\begin{align}
  \tauplus &\le \frac{\ai\obi + \sqrt{\ai \oai \bi \obi}}{2} , &
  \tauminus &\ge -\frac{\oai\obi + \sqrt{\ai \oai \bi \obi}}{2} .
\end{align}
By definition, $\rho\tauminus \le \thetarho \le \rho\tauplus$ and hence, $\thetarho \in [\thetaminus, \thetaplus]$, where
\begin{align}
  \thetaminus &\defas \max\left\{-\oai\obi, -\rho\frac{\oai\obi + \sqrt{\ai \oai \bi \obi}}{2} \right\} , &
  \thetaplus &\defas \min\left\{\ai\obi, \rho\frac{\ai\obi + \sqrt{\ai \oai \bi \obi}}{2} \right\} .
\end{align}
The function $\FktIV(\theta,\aii,\bii)$ is convex in $\theta$ by the concavity of entropy~\cite[Theorem~2.7.3]{Cover2006Elements} and consequently, $\mutInf[big]{\fOrig(\rvt x)}{\gOrig(\rvt y)} \le \max_{\theta \in \{\thetaplus,\thetaminus\}}\FktIV(\theta,\ai,\bi)$. Thus, \cref{thm:main} can be proved by establishing $1-\binEnt{\frac{\rho + 1}{2}}-\FktIV(\mytheta{},\ai,\bi) \ge 0$ for $\mytheta{} \in \{\thetaplus,\thetaminus\}$. Furthermore, it suffices to consider $\frac 12 < \ai < \bi < 1$ by continuity of $\FktIV$.

Define $\CST[\ai,\bi] \defas \frac{\ai\obi + \sqrt{\ai \oai \bi \obi}}{2}$, $\rhoI \defas \min\left\{\rho, \frac{\ai\obi}{\CST[\ai,\bi]}\right\}$, $\rhoII \defas \min\left\{\rho, \frac{\oai\obi}{\CST[\oai,\bi]}\right\}$, and
\begin{align}
  \FktI(\rho,\ai,\bi) \defas 1 - \binEnt{\frac{\rho + 1}{2}} - \FktIV(\rho \CST[\ai,\bi],\ai,\bi) .
\end{align}
Note that
\begin{align}
  \FktI(\rhoI,\ai,\bi) &= 1 - \binEnt{\frac{\rhoI + 1}{2}} - \FktIV(\thetaplus,\ai,\bi) \\
                       &\le 1 - \binEnt{\frac{\rho + 1}{2}} - \FktIV(\thetaplus,\ai,\bi)
\end{align}
by the monotonicity of the binary entropy function and accordingly we also have $\FktI(\rhoII,\oai,\bi) \le 1 - \binEnt{\frac{\rho + 1}{2}} - \FktIV(\thetaminus,\ai,\bi)$. \Cref{thm:main} thus follows from the following \lcnamecref{lem:main}.

\begin{lemma}
  \label{lem:main}
  For $0 < \aii < \bii < 1$ and $\rho \in \left[0, \frac{\aii \obii}{\CST}\right]$, we have $\FktI(\rho,\aii,\bii) \ge 0$ with equality if and only if $\rho = 0$.
\end{lemma}

Before proving \cref{lem:main}, we note the following facts.

\begin{lemma}
  \label{lem:basic_inequalities}
  For $x \in (0,1)$, we have
  \begin{align}
    \frac{1}{x^{-1}-1} + \log(1-x) > 0 .  
  \end{align}
\end{lemma}
\begin{proof}
  Using Taylor series expansion, we immediately obtain
  \begin{align}
    -\log(1-x) = \sum_{n=1}^\infty \frac{x^n}{n} < \sum_{n=1}^\infty x^n = \frac{x}{1-x} .
  \end{align}
                \end{proof}

The following \lcnamecref{lem:taylor} collects elementary facts about convex/concave functions and follows from elementary properties of convex functions on the real line (see, e.g., \cite[Chapter~I]{Roberts1974Convex}).
\begin{lemma}
  \label{lem:taylor}
  Let $\FktLemII \colon \setII \to \RR$ be a continuous function, defined on the compact interval $\setII \defas [\tI,\tII] \subset \RR$. Assuming that $f$ is twice differentiable on $\setIIb$, where $(\tI, \tII) \subseteq \setIIb \subseteq \setII$, the following properties hold.
  \begin{enumerate}
  \item \label{itm:taylor:positive}
    If $\FktLemII\pprm(\ti)\ge 0$ for all $\ti\in (\tI,\tII)$  
    and $\FktLemII\prm(\tii) = 0$ for some $\tii \in \setIIb$, then $\FktLemII(\ti) \ge \FktLemII(\tii)$ for all $\ti \in \setII$.
    Furthermore, if additionally $\FktLemII\pprm(\ti) > 0$ for all $\ti \in (\tI,\tII)$, then $\FktLemII(\ti) > \FktLemII(\tii)$ for all $\ti \in \setII\backslash\{\tii\}$.
  \item \label{itm:taylor:negative}
    If $\FktLemII\pprm(\ti)\le 0$ for all $\ti \in (\tI,\tII)$,
    then $\FktLemII(\ti) \ge \min\{\FktLemII(\tI),\FktLemII(\tII)\}$ for all $\ti \in \setII$.
    Furthermore, if $\FktLemII\pprm(\ti) < 0$ for all $\ti \in (\tI,\tII)$, then $\FktLemII(\ti) > \min\{\FktLemII(\tI),\FktLemII(\tII)\}$ for all $\ti \in (\tI,\tII)$.
  \end{enumerate}
\end{lemma}

\begin{proof}[Proof of \cref{lem:main}]
  Let $\setI \defas \{(\aii,\bii) \in \RR^2 : 0 < \aii < \bii < 1\}$, fix arbitrary $(\aii,\bii) \in \setI$ and define
  \begin{align}
     \rhommin &\defas \frac{\max\{\aii \bii, \oaii \obii \}}{\CST}, & \rhomin &\defas \frac{\min\{\aii \bii, \oaii\obii\}}{\CST}, & \rhomax &\defas \frac{\aii\obii}{\CST}. \label{eq:rhomindef}
  \end{align}
  We shall adopt the simplified notation $\FktIb(\rho) \defas \FktI(\rho,\aii,\bii)$, suppressing the fixed parameters $(\aii, \bii)$.
  For $\rho \in [0, \rhomax)$, we have the derivatives
  \begin{align}
    \FktIb\prm(\rho) &= \frac 12 \log_2\left(\frac{1+\rho}{1-\rho}\right) + \CST \log_2\left(\frac{(\oaii \bii - \CST\rho)(\aii \obii - \CST\rho)}{(\aii \bii + \CST\rho)(\oaii \obii + \CST\rho)}\right) , \label{eq:phi_firstderiv} \\
    \FktIb\pprm(\rho)                   &= \frac{\CST^2}{\log 2} \bigg(\frac{1}{\CST^2(1-\rho^2) } - \frac{1}{\oaii \bii - \CST\rho} \label{eq:d2} 
    \\*\nonumber&\qquad\qquad\qquad - \frac{1}{\aii\obii - \CST\rho} - \frac{1}{\oaii \obii + \CST\rho} - \frac{1}{\aii\bii + \CST\rho}\bigg) . 
  \end{align}
                      We write $\FktIb\pprm(\rho) = \frac{p(\rho)}{q(\rho)}$, where both $p$ and $q$ are polynomials in $\rho$, and choose
  \begin{align}
    q(\rho) &= \log(2) (1-\rho^2)(\oaii \bii - \CST\rho)
    \\*\nonumber&\qquad\qquad \times (\aii \obii - \CST\rho)(\oaii \obii + \CST\rho)(\aii\bii + \CST\rho) ,
  \end{align}
  such that $q(\rho) > 0$ for $\rho \in [0,\rhomax)$. 
  By \cref{eq:d2}, $p(\rho)$ is given by
  \begin{align}
    p(\rho) &= (\oaii \bii - \CST\rho)(\aii \obii - \CST\rho)(\oaii \obii + \CST\rho)(\aii\bii + \CST\rho) 
              \nonumber\\&\qquad - \CST^2(1-\rho^2)\Big( (\aii \obii - \CST\rho)(\oaii \obii + \CST\rho)(\aii\bii + \CST\rho)
    \nonumber\\&\qquad\qquad + (\oaii \bii - \CST\rho)(\oaii \obii + \CST\rho)(\aii\bii + \CST\rho)
    \nonumber\\&\qquad\qquad + (\oaii \bii - \CST\rho)(\aii \obii - \CST\rho)(\aii\bii + \CST\rho)
    \nonumber\\&\qquad\qquad + (\oaii \bii - \CST\rho)(\aii \obii - \CST\rho)(\oaii \obii + \CST\rho)
    \Big) .
    \label{eq:p}
  \end{align}
  This entails $\deg(p) \le 5$ and a careful
  calculation of the coefficients reveals $\deg(p) \le 3$.

  We will now demonstrate that there is a unique point $\rho^* \in (0,\rhomax)$, such that $p(\rho^*) = 0$.
  To this end, reinterpret $\FktIb\pprm(\rho)$ as a rational function of $\rho$ on $\RR$. We evaluate \cref{eq:p} and use $\aii < \bii$ to obtain the two inequalities
  \begin{align}
    p(0)
                  &= \aii \oaii \bii \obii \left( \aii \oaii \bii \obii - \CST^2 \right)
                                                                                                                         > 0 , \text{ and}\\
                                                                 p(\rhomax)           &= - \big(\CST^2- (\aii \obii)^2\big) (\bii - \aii) \obii \aii 
                                                                                                                                                      < 0 .
  \end{align}
  The number of roots of 
  $p$ in $(0,\rhomax)$ is thus odd and at most equal to its degree, i.e., either one or three.
    If we have $\rhomin \le 1$, then evaluation of \cref{eq:p} readily yields $p(-\rhomin) \le 0$. If, on the other hand, $\rhomin > 1$, we obtain $p(-\rhommin) \le 0$ from \cref{eq:p}. Thus, $p$ has at least one negative root and a unique root $\rho^* \in (0,\rhomax)$.
              \Cref{fig:sketch} qualitatively illustrate the behavior of $p(\rho)$ and $\FktIb\pprm(\rho)$.
  \begin{figure}[ht]
    \centering
    \begin{subfigure}[t]{0.48\textwidth}
      \centering
      \begin{tikzpicture}
  \newcommand{\ppoly}{0.1*(x-1.1)*(x+0.8)*(x-5)}
  \newcommand{\qpoly}{(0.8*(2-x)*(1.5+x))}
  \begin{axis}[sketch style,%
    xtick={-1.5,2},%
    xticklabels={$-\rhomin$, $\rhomax$},%
    ytick={0},%
    xmin=-3,%
    xmax=3,%
    ymin=-3,%
    ymax=3,%
    ]
    \addplot[myred,dashed,samples=100] {\ppoly};
    \addlegendentry{$p(\rho)$};
    \addplot[myblue,domain=-1.49:1.99,samples=100,id=plot_A1] function{\ppoly/\qpoly};
    \addplot[myblue,domain=2.01:,samples=100,id=plot_A2] function{\ppoly/\qpoly};
    \addplot[myblue,domain=:-1.51,samples=100,id=plot_A3] function{\ppoly/\qpoly};
    \addlegendentry{$\FktIb''(\rho)$};
  \end{axis}
\end{tikzpicture}
      \caption{$\rhomin < 1$}
      \label{fig:sketch1}
    \end{subfigure}%
    \begin{subfigure}[t]{0.48\textwidth}
      \centering
      \begin{tikzpicture}
  \pgfmathsetmacro\myrhoplus{0.5}
  \pgfmathsetmacro\myrhoplusright{\myrhoplus+0.01}
  \pgfmathsetmacro\myrhoplusleft{\myrhoplus-0.01}
  \newcommand{\nppoly}{0.05*(0.2-x)*(2.3+x)*(5-x)}
  \newcommand{\nqpoly}{(0.3*(1-x^2)*(2+x)*(3+x)*(\myrhoplus-x)*(4-x))}
  \begin{axis}[sketch style,%
    xtick={-3,-2,-1,\myrhoplus,1},%
    xticklabels={$-\rhommin$, $-\rhomin$, $-1$, $\rhomax$, $1$},%
    ytick={0},%
    xmin=-3.5,%
    xmax=1.5,%
    ymin=-0.8,%
    ymax=0.8,%
    ]
    \addplot[myred,dashed,samples=100] {\nppoly};
    \addlegendentry{$p(\rho)$};
    \addplot[myblue,domain=-3.5:-3.01,samples=100,id=plot_B1] function{\nppoly/\nqpoly};
    \addplot[myblue,domain=-2.99:-2.01,samples=100,id=plot_B2] function{\nppoly/\nqpoly};
    \addplot[myblue,domain=-1.99:-1.01,samples=100,id=plot_B3] function{\nppoly/\nqpoly};
    \addplot[myblue,domain=-0.99:\myrhoplusleft,samples=100,id=plot_B4] function{\nppoly/\nqpoly};
    \addplot[myblue,domain=\myrhoplusright:0.99,samples=100,id=plot_B5] function{\nppoly/\nqpoly};
    \addplot[myblue,domain=1.01:1.45,samples=100,id=plot_B6] function{\nppoly/\nqpoly};
    \addlegendentry{$\FktIb''(\rho)$};
  \end{axis}
\end{tikzpicture}
      \caption{$\rhomin > 1$}
      \label{fig:sketch2}
    \end{subfigure}
    \caption{Sketch of $p(\rho)$ and $\FktIb\pprm(\rho)$. }
    \label{fig:sketch}
  \end{figure}

  Consequently, $\FktIb\pprm(\rho) > 0$ for $\rho \in (0,\rho^*)$.
  By \scref{lem:taylor}{itm:taylor:positive}, $\FktIb(\rho) > \FktIb(0) = 0 $ for $\rho \in (0, \rho^*]$ as $\FktIb\prm(0) = 0$.
    Since $\FktIb\pprm(\rho) < 0$ for $\rho \in (\rho^*, \rhomax)$, we have $\FktIb(\rho) > \min\{\FktIb(\rho^*), \FktIb(\rhomax)\}$ for all $\rho \in (\rho^*, \rhomax)$, by~\scref{lem:taylor}{itm:taylor:negative}.
  In total, $\FktIb(\rho) > \min\{0,\FktIb(\rhomax)\}$ for $\rho \in (0,\rhomax)$.

  As $\FktIb(0) = 0$, it remains to show that $\FktI(\rhomax,\aii,\bii) > 0$ for $(\aii, \bii) \in \setI$. To this end, we introduce the transformation
  \begin{align}
    (\aii,\bii) \longmapsto (\cc,\xx) \defas  \left(\frac{\log\frac{\aii}{\bii}}{\log\frac{\aii\obii}{\oaii \bii}} \label{eq:xy}\, , \, \sqrt{\frac{\aii\obii}{\oaii \bii}} \right) ,
  \end{align}
  a bijective mapping from $\setI$ to $(0,1)^2$ with the inverse
  \begin{align}
    (\cc,\xx) \longmapsto (\aii,\bii) = \left( \frac{\xx^{2\cc}-\xx^2}{1-\xx^2} \ ,\,    \frac{1-\xx^{2-2\cc}}{1-\xx^{2}}\right) . \label{eq:ab}
  \end{align}
  In terms of $\cc$ and $\xx$, we have $\FktI(\rhomax,\aii,\bii) = \FktIII(\cc,\xx)$, where
  \begin{align}
    \FktIII(\cc,\xx) \defasAlign 1-\binEnt{\frac{1}{2} + \frac{\xx}{1 + \xx}} - \binEnt{\frac{\xx^{2\cc}-\xx^2}{1-\xx^2}} + \frac{1-\xx^{2-2\cc}}{1-\xx^2} \binEnt{\xx^{2\cc}} \\
    &= 1-\binEnt{\frac{1+3x}{2+2x}} + \frac{\binEnt{\xx^2}}{1-\xx^2} + \frac{\xx^{2\cc} \binEnt{\xx^{2-2\cc}} + \xx^{2-2\cc} \binEnt{\xx^{2\cc}}}{\xx^2-1} .
  \end{align}
  We fix a particular $\xx \in (0,1)$ and use the simplified notation $\FktV(\cc) \defas \FktIII(\cc,\xx)$, obtaining the derivatives
  \begin{align}
    \FktV\prm(\cc) &= \frac{2\log(\xx)}{(\xx^2-1)\log(2)} \bigg[ 2 \xx^{2\cc} \cc \log(\xx) 
    \\*\nonumber&\qquad\qquad + \xx^{2(1-\cc)} \log(1-\xx^{2\cc}) - \xx^{2\cc} \log(\xx^{2\cc} -\xx^2)\bigg] , \\
    \label{eq:log_ineq_applied} \FktV\pprm(\cc) &= \frac{4\log(\xx)^2 \xx^{2\cc}}{(1-\xx^2) \log(2)}
                   \Bigg[ \left(\frac{1}{\xx^{-2(1-\cc)}-1} + \log(1-\xx^{2(1-\cc)}) \right) 
    \\*\nonumber&\qquad\qquad\qquad\qquad\qquad + \frac{\xx^2}{\xx^{4\cc}} \left( \log(1-\xx^{2\cc}) + \frac{1}{\xx^{-2\cc}-1} \right) \Bigg] .
  \end{align}
  By applying \cref{lem:basic_inequalities} twice, we obtain $\FktV\pprm(\cc) > 0$. Thus, $\FktV(\cc) > \FktV(\frac 12)$ by~\scref{lem:taylor}{itm:taylor:positive} as $\FktV\prm(\frac 12) = 0$. It remains to show that $\FktVII(\xx) \defas \FktIII(\frac 12,\xx) > 0$. Note that $\FktVII(0) = \FktVII(1) = 0$ and
  \begin{align}
    \FktVII\prm(\xx) = \frac{1}{(1+\xx)^2} \log_2\big[ (1+3\xx) (1-\xx) \big] ,
  \end{align}
  for $\xx \in [0,1)$. If $\FktVII(\xx) \le 0$ for any $\xx \in (0,1)$ then $\fOrig$ necessarily attains its minimum in $(0,1)$ and there exists $\xx^* \in (0,1)$ with $\FktVII(\xx^*) \le 0$ and $\FktVII\prm(\xx^*) = 0$.
  As $\xx^* = \frac 23$ is the only point in $(0,1)$ with $\FktVII\prm(\xx^*) = 0$ and $\FktVII\left(\frac 23\right) = \log_2\left(\frac{27}{25}\right) > 0$, this concludes the proof.
\end{proof}

\section{Proof of \texorpdfstring{\Cref{pro:uniqueness}}{Proposition~\ref{pro:uniqueness}}}
\label{sec:proof:uniqueness}

We may assume $0<\rho<1$ and $\frac 12 \le a \le b \le 1$ by virtue of \cref{eq:symmetries}.
Clearly, $\gOrig=\pm \fOrig=\pm\bvt{i}$ for some $i \in \Ntoo{n}$ is a sufficient condition to
maximize $\mutInf[big]{\fOrig(\rvt x)}{\gOrig(\rvt y)}$. A careful inspection of the proof of \cref{thm:main} shows that this condition is also necessary.

In the following, we will use the notation of \cref{sec:main-result}.
As $b=1$ implies $\mutInf[big]{\fOrig(\rvt x)}{\gOrig(\rvt y)} = 0$, we assume $\frac 12 \le a \le b < 1$.
For equality in \cref{thm:main}, we need either $\phi(\rhoI, \ai, \bi) = 0$ or $\phi(\rhoII, \oai, \bi) = 0$.
By \cref{lem:main}, $\phi(\rhoII, \oai, \bi) > 0$ unless $\oai = \ai = \frac 12$, which in turn implies $\phi(\rhoII, \oai, \bi) = \FktI(\rhoI, \ai, \bi)$. The equality $\FktI(\rhoI, \ai, \bi) = 0$ can only occur for $\bi=\ai$, implying $\rhoI = \rho$. We want to show that $\FktI(\rho,\ai,\ai) = 0$ implies $\ai = \frac 12$. For $\ai \neq \frac 12$ we have
\begin{align}
  \frac{\partial \FktI}{\partial \rho} (\rho,\ai,\ai) &= \frac 12 \log_2\left(\frac{1+\rho}{1-\rho}\right) - \ai \oai \log_2\left(\frac{\rho}{\ai\oai \orho^2}+1 \right) , \label{eq:psi2_diff1} \\
  \frac{\partial^2 \FktI}{\partial \rho^2} (\rho,\ai,\ai) &= \frac{\rho(1-2\ai)^2}{\log(2) (\ai+\rho\oai)(1-\ai\orho)(1-\rho^2)} > 0 . \label{eq:t1_f2}
\end{align}
\sCref{lem:taylor}{itm:taylor:positive} now yields $0 = \FktI(0,\ai,\ai) < \FktI(\rho,\ai,\ai)$ as $\frac{\partial \FktI}{\partial \rho} (0,\ai,\ai) = 0$.
By the strict convexity of $\FktIV(\mytheta{}, \frac 12, \frac 12)$ in $\mytheta{}$, necessarily $\thetarho = \frac{\pair{\fOrig}{\Tr_{\rho} \gOrig}}{4} \in \{\thetaplus, \thetaminus\} = \pm \frac{\rho}{4}$. 
The Cauchy-Schwarz inequality together with \cite[Proposition~2.50]{ODonnell2014Analysis} yields $\rho^2 = \pair[big]{ \fOrig}{\Tr_{\rho} \gOrig}^2 = \pair[big]{\Tr_{\sqrt{\rho}} \fOrig}{\Tr_{\sqrt{\rho}} \gOrig} ^2 \le \pair{\fOrig}{\Tr_\rho \fOrig} \pair{\gOrig}{\Tr_\rho \gOrig} \le \rho^2$. Thus, necessarily $\gOrig=\pm \fOrig = \pm \bvt{i}$ for some $i \in \Ntoo{n}$ by \cite[Proposition~2.50]{ODonnell2014Analysis}.

\section{Discussion}
\label{sec:discussion}
The key idea underlying the proof of \cref{thm:main} is to split $\mytheta{1} = \tauplus + \tauminus$ into its positive and negative part (see \cref{sec:main-result}).
After reducing the problem to the inequality in \cref{lem:main}, the remaining proof is routine analysis. However, \cref{lem:main} might turn out to be useful in the context of other converse proofs, in particular for the optimization of rate regions with binary random variables.

\section*{Acknowledgment}
The authors would like to thank the anonymous referee for very helpful comments, that greatly improved the readability of the paper.

\bibliographystyle{imsart-number}
\bibliography{IEEEabrv,MonthAbr,literature}

\label{page:last-page}
\end{document}